\documentclass[a4paper,UKenglish,cleveref,autoref,thm-restate,pdfa]{lipics-v2021}

\title{Temporal Explorability Games} 

\titlerunning{Temporal Explorability Games}

\author{Pete Austin}{University of Liverpool, United Kingdom}{sgpausti@liverpool.ac.uk}{https://orcid.org/0000-0003-0238-8662}{}
\author{Sougata Bose}{UMONS – Universit\'e de Mons, Belgium}{sougata.bose@umons.ac.be}{https://orcid.org/0000-0003-3662-3915}{}
\author{Nicolas Mazzocchi}{Slovak University of Technology in Bratislava, Slovak Republic}{nicolas.mazzocchi@stuba.sk}{https://orcid.org/0000-0001-6425-5369}{}
\author{Patrick Totzke}{University of Liverpool, United Kingdom}{totzke@liverpool.ac.uk}{https://orcid.org/0000-0001-5274-8190}{} 

\authorrunning{P.~Austin, S.~Bose, N.~Mazzocchi, and P.~Totzke} 

\Copyright{Pete Austin, Sougata Bose, Nicolas Mazzocchi and Patrick Totzke}



\category{} 

\relatedversion{} 

 \supplement{}

\funding{
	This work has been supported by the EPSRC, projects EP/V025848/1 and EP/X042596/1 and by the
	Fonds de la Recherche Scientifique – FNRS under Grant n° T.0188.23
	(PDR ControlleRS).
	The research was mainly performed when N.~Mazzocchi was affiliated with the Institute of Science and Technology Austria where he was financed by the ERC-2020-AdG 101020093.	
}

\nolinenumbers 

\hideLIPIcs  

\EventEditors{John Q. Open and Joan R. Access}
\EventNoEds{2}
\EventLongTitle{42nd Conference on Very Important Topics (CVIT 2016)}
\EventShortTitle{CVIT 2016}
\EventAcronym{CVIT}
\EventYear{2016}
\EventDate{December 24--27, 2016}
\EventLocation{Little Whinging, United Kingdom}
\EventLogo{}
\SeriesVolume{42}
\ArticleNo{23}
\usepackage{tikz}
\usepackage[basic]{complexity}
\usepackage{colortbl}
\usepackage{mathtools}
\usepackage{amsmath}
\usepackage{dsfont}
\usepackage{centernot}
\usepackage{todonotes}
\colorlet{NMColor}{brown!60}
\newcommand{\notabene}[2][]{{\todo[color=NMColor, size=\footnotesize,#1]{\normalcolor\normalfont NM: #2}}}
\newcommand{\NMColor}[1]{\textcolor{NMColor}{#1}}
\newcommand{\ptnote}[2][]{\todo[color=blue!25,#1]{PT: #2}}

\newcommand{\panote}[2][]{\todo[color=green!25,#1]{PA:#2}}

\usetikzlibrary{positioning}
\usetikzlibrary{arrows}
\usetikzlibrary{calc}
\usetikzlibrary{shapes}
\usetikzlibrary{automata}

\usetikzlibrary{decorations.pathreplacing}
\usetikzlibrary{decorations.text}

\usetikzlibrary{fit}
\usetikzlibrary{backgrounds}
\usetikzlibrary{calc}
\usetikzlibrary{arrows.meta}

\usetikzlibrary{trees}
\usetikzlibrary{shapes}
\usetikzlibrary{decorations.pathmorphing}
\usetikzlibrary{decorations.markings}
\usetikzlibrary{arrows}

\usetikzlibrary{arrows.meta}
\usetikzlibrary{shapes.geometric}
\usetikzlibrary{fit}
\tikzset{>=stealth, shorten >=1pt}

\tikzset{every state/.style={
            very thick,
            fill=black!10,
            rounded corners=1mm,
            minimum size=6mm,inner sep=2pt,
}}

\tikzset{anode/.style={
        font=\small,
        rectangle
}}
\tikzset{ostate/.style={state,rectangle}}
\tikzset{estate/.style={state,
            minimum size=8mm,
diamond}}

\tikzset{every picture/.style={
            >={Triangle},
            thick,
            ->,
            initial text={},
            initial distance=2.5em,
            thick
}}

\tikzset{every edge/.append style={
  shorten <=2pt,
  shorten >=2pt,
}}
\tikzset{every loop/.append style={
  shorten <=2pt,
  shorten >=2pt,
}}

\tikzset{every label/.style={
        text=black!50,
        font=\small,
}}

\tikzset{
         strike through/.append style={
    decoration={markings, mark=at position 0.5 with {
    \draw[-] ++ (-5pt,-5pt) -- (5pt,5pt);}
  },postaction={decorate}}
}

\tikzstyle{zoki state}=[thick,minimum size=18pt, circle,draw, font=\small]
\tikzstyle{zoki estate}=[thick,minimum size=22pt, diamond,draw, font=\small]
\tikzstyle{zoki transition}=[->,thick,>=stealth,shorten >=2pt,shorten <=2pt, font=\small]
\tikzstyle{loop above right}=[out=60,in=30, min distance=5mm, looseness=8]
\tikzstyle{loop above left}=[out=150,in=120, min distance=5mm, looseness=8]
\tikzstyle{loop below left}=[out=-120,in=-150, min distance=5mm, looseness=8]
\tikzstyle{loop below right}=[out=-30,in=-60, min distance=5mm, looseness=8]
\newcommand{\zokiInitialLength}{18pt}
\newcommand{\zokiInitialStyle}{solid}
\newcommand{\zokiInitialAngle}{180}
\newcommand{\zokiInitialPos}{left}
\newcommand{\zokiInitialText}{}
\tikzoption{zoki initial style}{\tikzaddafternodepathoption{\renewcommand{\zokiInitialStyle}{#1}}}
\tikzoption{zoki initial length}{\tikzaddafternodepathoption{\renewcommand{\zokiInitialLength}{#1}}}
\tikzoption{zoki initial angle}{\tikzaddafternodepathoption{\renewcommand{\zokiInitialAngle}{#1}}}
\tikzoption{zoki initial text}{\tikzaddafternodepathoption{\renewcommand{\zokiInitialText}{#1}}}
\tikzoption{zoki initial pos}{\tikzaddafternodepathoption{\renewcommand{\zokiInitialPos}{#1}}}
\tikzstyle{zoki initial}=[after node path={{%
		[to path={[zoki transition] (\tikztostart) --  (\tikztotarget) \tikztonodes}]%
		(\tikzlastnode.\zokiInitialAngle)++(\zokiInitialAngle:\zokiInitialLength) edge[at start, \zokiInitialStyle] node[\zokiInitialPos, inner sep=0pt, outer sep=1pt] {\zokiInitialText} (\tikzlastnode)%
}}]
\tikzstyle{zoki final}=[after node path={ node[zoki state, scale=.8] at (\tikzlastnode) {} }]

\newcommand{\N}{\mathbb{N}}

\renewcommand{\epsilon}{\varepsilon}

\newcommand{\x}{\times}

\newcommand{\abs}[1]{\lvert #1 \rvert}

\newcommand{\eqby}[2][=]{\stackrel{\text{{\tiny{#2}}}}{#1}}
\newcommand{\eqdef}{\eqby{def}}

\newcommand{\arena}{\mathcal{A}}

\newcommand{\step}[2][]{\Step{#2}{}{#1}}
\newcommand{\Step}[3]{\ensuremath{\,{\stackrel{#1}{\longrightarrow}}{}^{\scriptstyle{#2}}_{\scriptstyle{#3}}}\,}

\newcommand{\states}{V}
\newcommand{\nodes}{V}
\newcommand{\edges}{E}

\newcommand{\true}{\mathop{true}}

\newcommand{\hor}{h}

\newcommand{\init}{s_0}
\newcommand{\temp}{\theta}

\RequirePackage{xspace}
\newcommand{\PONE}{Player~1\xspace}
\newcommand{\PTWO}{Player~2\xspace}

\begin{document}

\maketitle

\begin{abstract}
	Temporal graphs extend ordinary graphs with discrete time that affects the availability of edges.
We consider solving games played on temporal graphs where one player aims to explore the graph, i.e., visit all vertices.
The complexity depends majorly on two factors: the presence of an adversary and how edge availability is specified.

We demonstrate that on static graphs, where edges are always available, solving explorability games is just as hard as solving reachability games.
In contrast, on temporal graphs, the complexity of explorability coincides with generalized reachability ($\NP$-complete for one-player and $\PSPACE$-complete for two player games).
We show that if temporal graphs are given symbolically, even one-player reachability (and thus explorability and generalized reachability) games are $\PSPACE$-hard.
For one player, all these are also solvable in $\PSPACE$ and for two players, they are in $\PSPACE$, $\EXP$ and $\EXP$, respectively.

%
%

\end{abstract}

\section{Introduction}
\label{sec:intro}
Two player zero-sum games on graphs are a common formalism in formal verification, especially reactive synthesis (see e.g.~\cite{GAMES2002,GAMES2023,PR1989,P1977}).
The two players jointly determine an infinite path stepwise, where the owner of the current vertex gets to extend the path to a valid successor.
One player aims to satisfy a given winning condition, such as reaching a target vertex, and the opposing player aims to prevent that.
Under very general conditions, that are satisfied here, such games are determined, meaning that exactly one player has a winning strategy that guarantees a favourable outcome for them no matter their opponent's choices.
Solving a game refers to the algorithmic task to determine which player has a winning strategy.

Temporal graphs are a way to model dynamic systems: they extend graphs with discrete and global time and can specify for each edge at which times it can be traversed.
Temporal graphs are typically given as sequence of graphs over the same vertex set, an encoding we refer to as \emph{explicit}.
For our purposes this is polynomially equivalent to encoding temporal graphs as ordinary graphs in which (directed) edges carry an explicit list of timestamps, which we also refer to as explicit encodings\footnote{This is trivial for timestamps in unary encoding; Allowing binary-encoded timestamps in the input does not impact the complexity of solving explorability games since winning plays necessarily use consecutive times and thus cannot exist if the edge relation is too sparse.}.
Alternatively, and more interestingly, we consider succinct symbolic representations where the availability of edges is given as logical predicate in the existential fragment of Presburger Arithmetic. 

In this paper we study the computational complexity of solving games played on temporal graphs (where players' choices are required to respect the temporal constraints on edges).
Our focus lies on \emph{explorability} games, where \PONE wins if and only if all vertices of the graph are visited.
Explorability generalizes reachability conditions in the sense that reachability games straightforwardly and in logarithmic space reduce to explorability games.
On the other hand, explorability is a special case of so-called \emph{generalized reachability} conditions of \cite{FH12}, where several target sets are given and at least one vertex must be reached in each.

\medskip
For the benefit of readers unfamiliar with temporal graphs, or turn-based games, we
start by providing some intuition to help appreciate the difficulties caused by the dynamic changes of the arena, as well as the presence of antagonistic choice.
Consider first the one-player explorability game played on the arena $\arena_1$ in
\cref{fig:initial-examples-b}. Starting in $s$, \PONE wins by exploring the graph as $s\to t\to u \to v$. From any other vertex there is no exploring path due to the non-availability of edges at time $0$.
An easy property that implies non-explorability of static graphs
is the presence of two pairwise non-reachable vertices.
Indeed, if two such vertices exists then any path can at most see one of them, hence not explore the graph.
This can be turned into a full characterization of explorability even in the presence of an opponent (see \cref{lem:co-reach}). However, due to the temporal constraints on edges, this cannot na\"ively be extended to temporal graphs. For example, arena $\arena_1$ is explorable yet $v$ is not reachable from $t$ starting at time 0, and vice versa.
On the other hand, in arena $\arena_2$ (\cref{fig:initial-examples-a}), for every pair of vertices at least one can reach the other starting at time $0$.
Yet, from no initial vertex and at no time, $\arena_2$ is explorable.

\Cref{fig:initial-examples-c} demonstrates that even if \PONE wins the explorability game (as here from vertex $s$), she may not be able to enforce visiting the vertices in a particular order, as that can be influenced by her opponent's moves.
Dually, as for the game starting in vertex $s$ in $\arena_2$, 
even if \PTWO wins, he may not be able to dictate which vertex is left unexplored.



\begin{figure}
	\begin{subfigure}[t]{0.3\textwidth}\centering
		\scalebox{1.15}{\begin{tikzpicture}[]
				
	\node[zoki state](u2) at (0,0){$u$};
	\node[zoki state](s2) at (0,2){$s$};
	\node[zoki state](t2) at (2,0){$t$};
	\node[zoki state](v2) at (2,2){$v$};
		
	\path[zoki transition]
	(u2) edge (s2)
	(s2) edge node[auto]{$0$} (t2)
	(t2) edge node[above]{$1$} (u2)
	(s2) edge node[above]{$0, 3$} (v2)
	(v2) edge [loop right] node[right]{} (v2)
	;
\end{tikzpicture}}
		\caption{Arena $\arena_1$: explorable from $s$, not from $t,u,v$.}
		\label{fig:initial-examples-b}
	\end{subfigure}
	\hfill
	\begin{subfigure}[t]{0.3\textwidth}\centering
		\scalebox{1.15}{\begin{tikzpicture}[]
				
	\node[zoki state](s1) at (0,0){$u$};
	\node[zoki state](t1) at (0,2){$s$};
	\node[zoki state](u1) at (2,0){$t$};
	\node[zoki state](v1) at (2,2){$v$};
		
	\path[zoki transition]
	(s1) edge node[above]{} (t1)
	(t1) edge node[above]{} (u1)
	(u1) edge node[above]{} (s1)
	(u1) edge node[auto,swap]{$0$} (v1)
	(t1) edge node[above]{$0,1$} (v1)
	(v1) edge [loop right] node[right]{} (v1)
	;
\end{tikzpicture}}
		\caption{Arena $\arena_2$: not explorable from any vertex.}
		\label{fig:initial-examples-a}
	\end{subfigure}
	\hfill
	\begin{subfigure}[t]{0.3\textwidth}\centering
		\scalebox{1.15}{\begin{tikzpicture}[]
				
	\node[zoki estate](s1) at (0,2){$s$};
	\node[zoki state](t1) at (2,2){$t$};
	\node[zoki state](u1) at (0,0){$u$};
	\node[zoki state](v1) at (2,0){$v$};
		
	\path[zoki transition]
		(s1) edge  (t1)
		(s1) edge  (u1)
		(u1) edge[bend left=8]  (v1)
		(v1) edge[bend left=8]  (u1)
		(t1) edge[bend left=8]  (v1)
		(v1) edge[bend left=8]  (t1)
	;
\end{tikzpicture}}
		\caption{Arena $\arena_3$: explorable from $s$ but the order is not enforceable.}
		\label{fig:initial-examples-c}
	\end{subfigure}
	\caption{Examples of games on temporal graphs.
		Vertices owned by \PONE are drawn as circles, those owned by \PTWO as diamonds. The labels on edges denote the times at which they are available. Edges without labels are always available.
	}\label{fig:example}
\end{figure}

\subparagraph{Related Work}
Temporal graphs have been used to analyse dynamic networks and distributed systems in dynamic topologies, such as dissemination/propagation of information~\cite{Chen23,DAE23,KLO2010,ravi1994} or the spread of diseases~\cite{WANG2022311}.
There is a large body of mathematical work that considers temporal generalizations of various graph-theoretic notions and properties~\cite{DFS2023,FMS2009,MCS2014}.
Related algorithmic questions include
graph colouring~\cite{MMZ21},
travelling salesman~\cite{MS14},
maximal matchings~\cite{MMN23},
and checking the existence of temporal cliques~\cite{M24},
Eulerian circuits~\cite{Bumpus21},
vertex-cover sets~\cite{AMGZ20}, and
explorability~\cite{Adamson22,AKRIDA2021179,EHK21} (called exploration).
Several prior works on explorability assume structural properties that ensure that the graph is explorable. Questions then focus on the minimal time to explore.
Pelc~\cite{Pelc18} provides several sufficient conditions on temporal graphs to be explorable.
Spirakis and Michail~\cite{MS14} showed that without assumptions on the graph, checking explorability can be done in linear time for static graphs and is $\NP$-complete for temporal graphs.
In most of the works, a path in the temporal graph is allowed to wait at a vertex for an unbounded amount of time.

The edge relation is often deliberately left unspecified and sometimes only assumed to satisfy some weak assumptions about connectedness, frequency, or fairness
to study the worst or average cases in uncontrollable environments.
Depending on the application, one distinguishes between ``online'' questions (e.g.~\cite{FGJ24,MMS11}), where the edge availability is revealed stepwise, as opposed to the ``offline'' variant where all is given in advance.
We refer to~\cite{DRM12,HJ19,M2015} for overviews of temporal graph theory and its applications.

Fijalkow and Horn~\cite{FH12} study games with \emph{generalized reachability} objectives, which generalize not only reachability but also explorability games.
Generalized reachability conditions are conjunctions of reachability conditions:
$\PONE$ aims to reach at least one vertex out of each of several target sets.
Explorability thus corresponds to the case where every vertex forms a singleton target set.
Solving generalized reachability games is $\PSPACE$-complete~\cite{FH12}, but in polynomial time if all target sets are singletons.

Reachability and parity games played on (symbolically represented) temporal graphs have been introduced in~\cite{ABT24}; solving these games is~$\PSPACE$-complete. The lower bound was shown by reduction from the emptiness problem for unary alternating finite automata, crucially relying on the presence of antagonistic choice.
As in~\cite{ABT24}, our notion of temporal graphs assumes that some edge must be traversed at every unit time. 
Waiting (not moving the token for a while) as occasionally permitted \cite{M2015,AMGZ20} can be modelled by adding self-loops.

Turn-based games involving temporal constraints have also been studied in the context of games played on the configuration graphs of timed automata~\cite{AD1994}.
Solving timed parity games is complete for $\EXP$~\cite{CHP2011,MPS1995} and the lower bound already holds for reachability games on timed automata with only two clocks~\cite{JT2007}.
However, the time in temporal graphs is discrete as opposed to the continuous time semantics in timed automata.
A direct translation of (games on) temporal graphs to equivalent timed automata games requires two clocks: one to hold the global time used to
check the edge predicate and one to ensure that time progresses one unit per step.
Fearnley and Jurdzi\'nski \cite{FJ13} showed that the reachability problem (solving one-player reachability games) for two-clock timed automata
is already \PSPACE-hard.
Our lower bound constructions have a similar flavour but are incomparable in that they do not re-prove nor are implied by their result. They make crucial use of resetting of clocks which is impossible in our model. Our construction in turn uses either antagonistic choice (\cref{thm:temp-2p}) or the more powerful transition guards in symbolic representations (\cref{lem:symbolic-hard}).

\begin{figure}[t]
	\centering
	{
		\resizebox{1.05\textwidth}{!}{
			\begin{tabularx}{15cm}{ccccc}
				                                                      &                                                               & \textbf{Static}           & \textbf{Explicit}           & \textbf{Symbolic}                                       \\
				\cline{2-5}
				                                                      & \multicolumn{1}{|c|}{\multirow{2}{*}{\textbf{Reachability}}}  & $\NL$-complete            & $\NL$-complete              & \multicolumn{1}{c|}{\textbf{$\PSPACE$-complete}}        \\
				                                                      & \multicolumn{1}{|c|}{}                                        & \cite[Theorem 4.18]{AB09} & [\cref{rmk:temp-complete}]  & \multicolumn{1}{c|}{[\cref{lem:symbolic-all}]}          \\
				                                                      & \multicolumn{1}{|c|}{\multirow{2}{*}{\textbf{Explorability}}} & \textbf{$\NL$-complete}   & \textbf{$\NP$-complete}     & \multicolumn{1}{c|}{\textbf{$\PSPACE$-complete}}        \\
				                                                      & \multicolumn{1}{|c|}{}                                        & [\cref{thm:static}]       & [\cref{thm:temp-1p}]        & \multicolumn{1}{c|}{[\cref{lem:symbolic-all}]}          \\
				                                                      & \multicolumn{1}{|c|}{\multirow{2}{*}{\textbf{Gen. Reach}}}    & $\NP$-complete            & $\NP$-complete              & \multicolumn{1}{c|}{\textbf{$\PSPACE$-complete}}        \\
				                                                      & \multicolumn{1}{|c|}{}                                        & \cite[Theorem 3]{FH12}    & [\cref{rmk:temp-complete}]  & \multicolumn{1}{c|}{[\cref{lem:symbolic-all}]}          \\
				\cline{2-5}
				\multirow{-8}{*}{\rotatebox{90}{\textbf{One-player}}} & \multicolumn{1}{|c|}{\multirow{2}{*}{\textbf{Reachability}}}  & $\P$-complete             & $\P$-complete               & \multicolumn{1}{c|}{$\PSPACE$-complete}                 \\
				                                                      & \multicolumn{1}{|c|}{}                                        & \cite{GAMES2002}          & [\cref{rmk:temp-complete}]  & \multicolumn{1}{c|}{\cite[Theorem 2]{ABT24}}            \\
				                                                      & \multicolumn{1}{|c|}{\multirow{2}{*}{\textbf{Explorability}}} & \textbf{$\P$-complete}    & \textbf{$\PSPACE$-complete} & \multicolumn{1}{c|}{\textbf{$\PSPACE$-hard; In $\EXP$}} \\
				                                                      & \multicolumn{1}{|c|}{}                                        & [\cref{thm:static}]       & [\cref{thm:temp-2p}]        & \multicolumn{1}{c|}{[\cref{lem:symbolic-all}]}          \\
				                                                      & \multicolumn{1}{|c|}{\multirow{2}{*}{\textbf{Gen. Reach}}}    & $\PSPACE$-complete        & $\PSPACE$-complete          & \multicolumn{1}{c|}{\textbf{$\PSPACE$-hard; In $\EXP$}} \\
				\multirow{-6}{*}{\rotatebox{90}{\textbf{Two-player}}} & \multicolumn{1}{|c|}{}                                        & \cite[Theorem 1]{FH12}    & [\cref{rmk:temp-complete}]  & \multicolumn{1}{c|}{[\cref{lem:symbolic-all}]}          \\
				\cline{2-5}
			\end{tabularx}
		}
	}
	\caption{A table detailing the computational complexities for the one and two-player variants of reachability, explorability, and generalized reachability games
		played on static, explicitly and symbolically represented temporal graphs. 
		New results are in boldface.
	}
	\label{fig:overview}
\end{figure}

\subparagraph{Contributions}
We study the complexity of solving explorability games and contrast the worst-case complexity for related decision questions for games played on static graphs versus temporal graphs.
It turns out that explorability is no harder than reachability on static graphs whereas on temporal graphs, 
explorability games exhibit the hardness of generalized reachability when the temporal edge availability is given explicitly. 
For temporal explorability games with a succinct, symbolic representation, we have $\PSPACE$-hardness for both variants but a $\PSPACE$-$\EXP$ complexity gap for two-player games.
Specifically,
\begin{enumerate}
	\item
	      on static graphs, solving explorability games is complete for polynomial time (and $\NL$-complete for one-player games);
	\item
	      on explicitly represented temporal graphs, explorability games are $\PSPACE$-complete ($\NP$-complete in the one-player variant);
	\item
	      reachability and thus explorability games on symbolically represented temporal graphs are $\PSPACE$-hard even in the single-player variant. This strengthens the known lower bound from \cite{ABT24} which required a second antagonistic player.
\end{enumerate}
{\Cref{fig:overview} summarizes these and related claims.}
Our most involved constructions are the $\PSPACE$~lower bounds, both from the satisfiability of quantified Boolean formulae (QBF).

\section{Notations}
\label{sec:preliminaries}


\begin{definition}[Temporal Graphs]
	A temporal graph $G=(\nodes,\edges)$ is a directed graph where $\nodes$ are vertices and $\edges:\nodes^2\to 2^\N$ is the edge availability relation
	that maps each pair of vertices to the set of times at which the respective directed edge can be traversed.
	If $i\in \edges(u,v)$ we call $v$ an \emph{$i$-successor} of $u$ and write $u\step{i}v$.
	We call $G$ \emph{static} if for all $u,v\in\nodes$,
$\edges(u,v)$ is either $\N$ or $\emptyset$.
	Its \emph{horizon} is $\hor(G)=\sup_{u,v\in\nodes}(\edges(u,v))$, that is, the largest finite time at which any edge is available, or $\infty$ if no such finite time exists.
\end{definition}

\begin{figure}
	\begin{center}
		\scalebox{1.00}{\begin{tikzpicture}[
	ORANGE/.style = {color=orange, fill=orange!20},
	RED/.style = {color=red, fill=red!20},
	BLUE/.style = {color=blue, fill=blue!20}
	]
	\node[zoki state] (s0) at (0, -1){$s,0$};
	\node[zoki state] (t0) at (0, -2){$t,0$};
	\node[zoki state] (u0) at (0, -3){$u,0$};
	\node[zoki state] (v0) at (0, -4){$v,0$};
	
	\node[zoki state] (s1) at (3, -1){$s,1$};
	\node[zoki state] (t1) at (3, -2){$t,1$};
	\node[zoki state] (u1) at (3, -3){$u,1$};
	\node[zoki state] (v1) at (3, -4){$v,1$};
	
	\node[zoki state] (s2) at (6, -1){$s,2$};
	\node[zoki state] (t2) at (6, -2){$t,2$};
	\node[zoki state] (u2) at (6, -3){$u,2$};
	\node[zoki state] (v2) at (6, -4){$v,2$};
	
	\node[zoki state] (s3) at (9, -1){$s,3$};
	\node[zoki state] (t3) at (9, -2){$t,3$};
	\node[zoki state] (u3) at (9, -3){$u,3$};
	\node[zoki state] (v3) at (9, -4){$v,3$};
	
	\node[zoki state] (s4) at (12, -1){$s,4$};
	\node[zoki state] (t4) at (12, -2){$t,4$};
	\node[zoki state] (u4) at (12, -3){$u,4$};
	\node[zoki state] (v4) at (12, -4){$v,4$};
	
	
	\draw (s0) edge node {}(t1);
	\draw (s0) edge node {}(v1);
	\draw (v0) edge node {}(v1);
	\draw (u0) edge node {}(s1);
	
	\draw (v1) edge node {}(v2);
	\draw (t1) edge node {}(u2);
	\draw (u1) edge node {}(s2);
	
	\draw (v2) edge node {}(v3);
	\draw (u2) edge node {}(s3);
	
	\draw (s3) edge node {}(v4);
	\draw (v3) edge node {}(v4);
	\draw (u3) edge node {}(s4);
	
\end{tikzpicture}}
		\caption{The expansion of $\mathcal{A}_1$.}
		\label{fig:expansion}
	\end{center}
\end{figure}
Naturally, one can unfold a temporal graph $G$ into its \emph{expansion}
up to time $T\in\N\cup\{\infty\}$,
which is the static graph $G_T$ with vertices $\nodes\x\{0,1,\ldots,T, T+1\}$
and $(u,\temp)\to (v,\temp+1)$ if and only if $\temp\in\edges(u,v)$.
We denote by \emph{the expansion} of a temporal graph $G$, its expansion up to its horizon $\hor(G)$. See for instance \cref{fig:expansion} for the expansion of a temporal graph $\mathcal{A}_1$ (up to its horizon $\hor(\?A_1)=3$).

\medskip
Complexity bounds for decision problems about temporal graphs very much depend on the representation of the input.
We will say a temporal graph is represented \emph{explicitly} if it is given as sequence of directed graphs, one per unit time from $0$ up to its (finite) horizon.
We also consider a \emph{symbolic} representation where the edge relation $\edges$
is represented as a formula in the existential fragment of Presburger Arithmetic ($\exists$PA), the first-order theory over natural numbers with equality and addition.
That is, the formula $\Phi_{u,v}(x)$ with one free variable $x$
represents the set of times at which an edge from $u$ to $v$ is available
as
\mbox{$\edges(u,v) = \{n \mid \Phi_{u,v}(n) \equiv \true\}$}.
We use common syntactic sugar including inequality and multiplication with constants.

In this representation, 
checking if an edge is available at a given time,
i.e., checking whether a given $\exists$PA formula is satisfied by a given valuation, 
is $\NP$-complete (and in polynomial time if the number of quantifiers are fixed) \cite[Corollary 1]{S84}.
The symbolic representation of a temporal graph can be exponentially more succinct than the explicit one: using repeated doubling one can express exponentially large values.
	{This representation is also at most exponentially larger because the Presburger-definable edge relation must be ultimately periodic with base and period at most exponential \cite{GS66}.}

\begin{definition}[Games on graphs]\label{def:games}
	A \emph{game} is played by two opposing players.
	It consists of
	a directed graph $(\states,\edges)$, a partitioning $\states=\states_1\uplus\states_2$ of the vertices into those controlled by \PONE and \PTWO respectively, and a \emph{winning condition}.
	We refer to $\arena=(\states_1,\states_2,\edges)$ as the \emph{arena} of the game.

	The game starts with a token on an initial vertex $\init\in\states$
	and is turn-based, where in round $j$, the owner of the vertex occupied by the token moves it to some successor. 
	This way an infinite path 
	$\rho=s_0s_1\ldots$
	called a \emph{play} is generated based on choices made by each player given the current round and vertex.
	A play is won by \PONE if it satisfies the given winning condition,
	and by \PTWO otherwise.

	A \emph{strategy} for Player~$i$ is a recipe for how to move.
	It is a function $\sigma_i:\states^*\states_i\to \states$
	from finite paths ending in a vertex $s$ in $V_i$ 
	to some successor.
	%
	A strategy is \emph{winning from $s\in\states$} if Player~$i$ wins every play
	that starts in $s$ and during which all decisions are according to $\sigma_i$.
\end{definition}

We call a vertex $s$ winning for Player~$i$ if there exists a winning strategy from $s$, and call the subset of all such vertices the \emph{winning region} for that player.
The main algorithmic question is to \emph{solve} a game, meaning in other words to compute the winning regions for a given arena and winning condition.
We consider the following winning conditions.
\begin{itemize}
	\item \emph{Reachability} towards a given set $F\subseteq \states$ of target vertices. This is satisfied by those plays that eventually visit a vertex in $F$.
	\item \emph{Generalized Reachability} towards target sets $F_1,F_2,\ldots,F_k\subseteq \states$. This is satisfied by those plays that visit at least one vertex of every target set $F_i$.
	\item \emph{Explorability} is the condition that asks that eventually
	      every vertex in $\states$ is visited.
\end{itemize}

Note that explorability is a special case of generalized reachability where every vertex
$s_i\in\states$ 
corresponds to one target set $F_i=\{s_i\}$.

\medskip
We study games played on the expansion of a given temporal graph
$G=(\states,\edges)$,
where the
ownership of vertices and winning condition are defined on the underlying temporal graph $G$ and lifted to \emph{the expansion} of $G$.
That is, for any time $\theta$ and $u\in \states_i$, the vertex $(u,\theta)$ is owned by Player~$i$.
Similarly, (generalized) reachability conditions are defined in terms of target sets $F\subseteq \states$ and explorability asks to visit every vertex in $\states$.

\section{Static Graphs}
\label{sec:static}
We first discuss the bounds of explorability games for both one and two player variants, on static graphs, i.e., every edge is available at any time.

\begin{lemma}\label{lem:reach-to-explore}
	For every arena $\arena = (\states_1,\states_2,\edges)$ and vertices $s,t\in V$
	one can construct (in logarithmic space) an arena $\?B= (\states'_1,\states'_2,\edges')$ such that
	$\states' = \states'_1 \cup \states'_2$, 
	$\states\subseteq \states'$,
	$\states_2=\states_2'$ and
	\PONE wins the reachability game on $\arena$ from $s$ to $t$ if, and only if,
	she wins the explorability game on $\?B$.
\end{lemma}

\begin{proof}
	Without loss of generality, assume that $t\in\states_1$ and that each vertex has at least one outgoing edge.
	We construct $\mathcal{B}$ as follows. Firstly, for every edge between vertices $v,u$, introduce a new vertex $[v,u]\in\states_1'$  that can either move to $u$ or reset, i.e.\ move back to $s$.
	Secondly, from target $t$ there exists an edge to every other vertex.

	We claim that \PONE wins the reachability game on $\arena$ if and only if she wins the explorability game on $\mathcal{B}$.
	Indeed, if \PONE does not win on $\arena$, then she cannot visit $t$ on either arena and therefore loses in both.
	Conversely, if \PONE does win on $\arena$, then she can explore $\mathcal{B}$ by repeatedly moving from $s$ to $t$; visit a previously unseen vertex, then reset to $s$.
	Note that a reduction from a one-player reachability game where \PONE owns all the vertices will construct another one-player explorability game where this still holds.
\end{proof}

The following characterization of explorability games is adapted from \cite[Theorem 4]{FH12},
where $s \preceq t$ denotes that \PONE wins the reachability game from $s$ with target $t$.

\begin{lemma}\label{lem:co-reach}
	Consider an arena with vertex set $\states$ and let $s\in\states$ be an initial vertex.
	\PONE wins the explorability game from $s$ if, and only if, both
	1) for all $u,v\in\states$ either $u\preceq v$ or $v\preceq u$; and
	2) for all $u\in\states$, $s\preceq u$.
\end{lemma}

\begin{proof}

	Suppose both conditions 1 and 2 are true.
	Point 1 implies that there exists a linearization of all $n$ vertices such that $v_1\preceq v_2\preceq\ldots\preceq v_n$.
	%
	By Point 2 we have that $s\preceq v_1$ and therefore there is such a linearization of $\preceq$ with $v_1=s$.
	Consequently, for all $1\leq i < n$ there exists a \PONE strategy $\sigma_i$
	that wins the reachability game from  $v_i$ to $v_{i+1}$.
	\PONE can follow the $\sigma_i$ from $v_i$ until $v_{i+1}$ is reached, then switch to $\sigma_{i+1}$ and so on, until all vertices have been visited.

	Suppose we have an arena where the first condition is false: there
	are $u,v\in \states$ with $u\not\preceq v$ and $v\not\preceq u$.
	Regardless of the starting vertex $s$, once a play visits $u$, \PONE cannot ensure to visit $v$ from then on (or vice versa).
	Finally, if the second condition is false
	then there must be one vertex $u$ that \PONE cannot ensure to visit from $s$.
\end{proof}

\begin{theorem}\label{thm:static}
	Solving one-player (respectively two-player) explorability games on a static graph is complete for $\NL$ (respectively $\P$).
\end{theorem}

\begin{proof}
	The hardness is a consequence of \cref{lem:reach-to-explore}.
	Note that the explorability game constructed from a one-player reachability game is also one-player.
	As one and two-player reachability games are $\NL$-hard and $\P$-hard respectively, the lower bounds follow.
	The upper bounds follow from \cref{lem:co-reach},
	observing that at most
	$n^2$ reachability queries are necessary to verify the two conditions in the lemma.
	These queries can be done in $\NL$ in the one-player case, and $\P$ in the two-player case.
\end{proof}

\section{Explicitly represented Temporal Graphs}
\label{sec:temporal}
Before discussing explorability games, 
we first remark that solving reachability and generalized reachability games have the same complexity on explicitly represented temporal graphs and on static graphs.
This is because for any temporal graph $\arena$ one can construct its expansion $\arena'$,
which is only polynomially larger assuming explicit encodings,
and modify the winning conditions appropriately to get a game on this static graph:
every target vertex $v$ in $\arena$ gives rise to target vertices $(v,\temp)$ in $\arena'$ for all times $0\le \temp \le \hor(\arena)$.

This reduction allows transferring complexity upper bounds for solving games on static arenas \cite{FH12}.

\begin{theorem}\label{rmk:temp-complete}
	Assuming explicit encodings,
	solving one-player games on temporal graphs is \NL-complete
	for reachability conditions and \NP-complete for generalized reachability.
	Solving two-player games on temporal graphs is \P-complete for reachability and
	\PSPACE-complete for generalized reachability.
\end{theorem}

Note that an explorability game on a temporal graph corresponds to a generalized reachability game on its expansion as outlined above.
However, the target sets are no longer singleton sets and therefore the improved (\NL. resp.\ \P) upper bounds provided in \cite{FH12} for (one resp.\ two-player) generalized reachability with singleton targets do not apply to explorability on temporal graphs.
For example,
the explorability game on $\mathcal{A}_1$ (\cref{fig:example}) corresponds to
the generalised reachability game
on its expansion (\cref{fig:expansion})
with targets
$\{s\}\x\underline{4}$,
$\{t\}\x\underline{4}$,
$\{u\}\x\underline{4}$, and
$\{v\}\x\underline{4}$, where $\underline{4}=\{0,1,2,3,4\}$ are all possible times up to the horizon.

In contrast to games on static arenas, where explorability is not harder than reachability, we will see that on temporal graphs explorability is as hard as generalized reachability. 

\begin{theorem}\label{thm:temp-1p}
	Solving one-player explorability games on an explicit temporal graph is \NP-complete.
\end{theorem}
\begin{proof}
	The $\NP$-hardness is due to Michail and Spirakis \cite[Proposition 2]{MS14}
	by reduction from the Hamiltonian Path problem.
	Indeed, a directed graph with $n$ vertices has a Hamiltonian path if and only if it can be explored in exactly $n$ steps.
	A matching upper bound can be achieved by stepwise guessing an exploring path
	of length at most $\hor(G)$, which is polynomial given that the input graph is explicitly represented.
	%
\end{proof}




We now consider the impact of an antagonistic player for explorability on temporal graphs.
A key element of our lower bound construction is the interaction of antagonistic choice and the passing of time. Our reduction, from QBF, relies on a time-bounded final "flooding phase" that allows exploration only if sufficient time is left, and which can only be guaranteed by winning the preceding QBF game.

\begin{theorem}\label{thm:temp-2p}
	Solving two-player explorability games on explicit temporal graphs is $\PSPACE$-complete.
\end{theorem}


\begin{proof}
	The upper bound for solving explorability games on a temporal arena $\arena$ follows from solving the generalized reachability game on the expansion of $\arena$, where we have a target set $F_i=\{(v_i,\temp)\mid \forall ~\temp\in \{0,1,\dots,\hor(\arena)\}\}$ for every vertex $v_i$ of $\arena$. Generalized reachability games can be solved in \PSPACE~\cite[Theorem 1]{FH12}.

	We provide a matching lower bound using a reduction from QSAT.
	Given a QBF formula
	$\Phi = \exists x_1\forall x_2\ldots \exists x_n.~C_1\land C_2\land\ldots\land C_k$, we construct a two-player explorability game $G_{\Phi}$ on a temporal graph so that \PONE wins iff $\Phi$ is true.
	Let $\varphi$ refer to the matrix $C_1\land C_2\land\ldots\land C_k$ of the given formula $\Phi$, that is, each $C_i$ is a disjunction of at most $3$ literals.
	The game $G_{\Phi}$
	has vertices $\states = \states_{\forall} \uplus \states_{\exists}\uplus \{q_{\varphi}\} \uplus \states_{\text{literals}} \uplus \states_{\text{clause}}$, where
	\begin{itemize}
		\item{
		      the vertices $\states_\exists = \{q_1,q_3,\ldots,q_n\}$,
		      $\states_{literals} = \{q_{x_1}, q_{\lnot x_1}, q_{x_2}, q_{\lnot x_2},\ldots, q_{x_n}, q_{\lnot x_n}\}$ and
		      $\states_{clauses} = \{q_{C_1}, q_{C_2}, \ldots, q_{C_k}\}$ are controlled by \PONE.
		      }
		\item{
		      the vertices	$\states_\forall \{q_2,q_4,\ldots,q_{n-1}\}$ and $\{q_\varphi\}$
		      are controlled by \PTWO.}
	\end{itemize}

	The game is played in $4$ phases, each of a fixed length, as follows.

	\begin{enumerate}
		\item
		      In the \emph{initial phase}, from time $0$ to $k-1$, the game starts from $q_{C_1}$ and visits all clause vertices of $\states_{clauses}$ in order and then ends up in the vertex $q_1$ corresponding to the first existential quantifier.
		      Formally, the edges available during this phase are, for all $0\le i<k$, $E(q_{C_i}, q_{C_{i+1}})= E(q_{C_k},q_1)= [0,k)$.
		      Edges of the initial phase are drawn in red in the example in \cref{fig:2p-temp-hard}.

		\item
		      The \emph{selection phase}, from time $k$ to $k+2n-1$, the game visits the vertices corresponding to the quantifiers in order of their appearance in the formula.
		      Depending on whether the quantifier is $\exists$ (or $\forall$), \PONE (respectively \PTWO) chooses to visit a vertex corresponding to either a positive or negative literal for each variable.
		      The edges available during this phase are for all $0\leq i <n$, $E(q_i,q_{x_i})=E(q_i,q_{\lnot x_i})=E(q_{x_i},q_{i+1})= E(q_{\lnot x_i},q_{i+1})=E(q_{x_n},q_{\varphi})= E(q_{\lnot x_{n}},q_{\varphi})=[k,k+2n)$.
		      Corresponding edges are in blue in \cref{fig:2p-temp-hard} and corresponds to choosing a valuation for the variables in the formula $\Phi$.

		\item
		      The \emph{evaluation phase}, at time $k+2n$ and $k+2n+1$, starts from vertex $q_{\varphi}$, where \PTWO selects a clause $C_i$ by moving to vertex $q_{C_i}$.
		      Formally, we have for all $0<i\leq k$, $E(q_{\varphi},q_{C_i})=\{k+2n\}$.
		      From $q_{C_i}$, \PONE has the choice to visit a literal vertex $q_{\ell}$ at time $k+2n+1$ if $\lnot \ell$ is contained in the clause $C_i$, where $\ell$ is a literal.
		      After this phase, all vertices in $\states_{\forall} \uplus \states_{\exists}\uplus \{q_{\varphi}\} \uplus \states_{\text{clause}}$ have been visited as well as half of $\states_{literal}$.
		      In \cref{fig:2p-temp-hard}, edges for the clause selection and literal selection are in green.

		\item
		      The game ends after a \emph{flooding phase} that lasts for exactly $n-1$ steps starting from time $k+2n+2$.
		      That is, all edges become unavailable from time $k+3n+1$ onwards.
		      During this phase, an edge is available from a literal vertex $q_{\ell_i}$ to both literal vertex $q_{x_{i+1}}$ and $q_{\lnot x_{i+1}}$.
		      Formally, the edges $(q_{\ell_i},q_{\ell_{i+1}})$, for all $0<i<n$ and $(q_{\ell_n},q_{\ell_1})$, where $\ell_j\in \{x_j,\lnot x_j\}$, are available at times $[k+2n+2,k+3n+1]$.
		      Thus, in this phase \PONE can visit exactly $n-1$ of the possibly unexplored vertices.
		      The flooding phase is shown in yellow in \cref{fig:2p-temp-hard}.
	\end{enumerate}

	Suppose the QBF formula $\Phi$ is true.
	Then the $\exists$-player has a strategy to assign values to the variables $x_{i}$ (for odd $i$) to ensure that $\varphi(\vec{x})$ is true.
	By following the same choices in the selection phase
	(moving $q_i\step{}q_{x_i}$ if $x_i$ is set to $\true$ and $q_i\step{}q_{\lnot x_i}$ otherwise),
	\PONE guarantees that when $q_\varphi$ is reached, the jointly chosen variable assignment $\vec{x}$
	satisfies $\varphi$.
	At this point the play has already visited exactly one vertex among $q_{x_j}$ and $q_{\lnot x_j}$ for all $j\le n$. Moreover, no matter which vertex $q_{C_i}$ is chosen by \PTWO in the next step, the corresponding clause $C_i$ is also satisfied by assignment $\vec{x}$.
	By construction, this means there exists a literal $\ell_j\in C_i$ such that vertex $q_{\ell_j}$
	has been visited before.
	The \PONE can then move $C_i\step{} q_{\lnot \ell_j}$ to end the evaluation phase.
	Finally, in the flooding phase, \PONE can freely move to visit the remaining $n-1$ literal vertices not visited thus far.
	This guarantees that indeed, all vertices have been visited and \PONE wins the explorability game.

	If the QBF instance $\Phi$ is not satisfiable, then the $\forall$-player has a strategy to assign values to the $x_{2i}$ to ensure that $\varphi(\vec{x})$ is false.
	By following the same choices
	during the selection phase
	(moving $q_i\step{}q_{x_i}$ if $x_i$ is set to $\true$ and $q_i\step{}q_{\lnot x_i}$ otherwise),
	\PTWO guarantees that when $q_\varphi$ is reached, the jointly chosen variable assignment
	does not satisfy $\Phi$, meaning that there is some clause $C_i$ that is not satisfied. In the evaluation phase, the strategy $\tau$ then picks $q_{C_i}$ as the successor from $q_{\varphi}$.
	Consequently, all successors
	of $q_{C_i}$ must be vertices that have already been seen in the play.
	This means that exactly $n$ of the literal vertices are not yet visited.
	However, the game ends after the following $n-1$ of the flooding phase.
	Therefore, one of the literal vertices must be left unexplored.
	%
\end{proof}

\begin{figure}[t]\centering
	\scalebox{1.15}{
		\begin{tikzpicture}[node distance=1.5cm]
			\clip (-.9,-2.8) rectangle (11.3,2.3);
			\tikzset{every label/.style={text=black, font=\small}}
			\node[zoki state, label=center:$q_1$] (q1) at (0,0) {};
			\node[zoki state, label=center:$q_{x_1}$, above right=of q1] (x1) {};
			\node[zoki state, label=center:$q_{\scalebox{.5}{$\lnot$} x_1}$, below right=of q1] (x1c) {};
			\node[zoki state, draw=none, label=center:$q_2$, below right=of x1] (q2)  {};
			\node[zoki estate] at (q2) {};
			\node[zoki state, label=center:$q_{x_2}$, above right=of q2] (x2) {};
			\node[zoki state, label=center:$q_{\scalebox{.5}{$\lnot$} x_2}$, below right=of q2] (x2c) {};
			\node[zoki state, opacity=0, below right=of x2] (qdots) {};
			\node[zoki state, label=center:$q_{x_3}$, above right=of qdots] (xn) {};
			\node[zoki state, label=center:$q_{\scalebox{.5}{$\lnot$} x_3}$, below right=of qdots] (xnc) {};
			\node[zoki state, draw=none, label=center:$q_{\varphi}$, below right=of xn] (formula) {};
			\node[zoki estate] at (formula) {};
			\node[zoki state, zoki initial, zoki initial angle=170, label=center:$q_{C_1}$, above right=of formula, yshift=8pt] (C1) {};
			\node[zoki state, label=center:$q_{C_4}$, below right=of formula, yshift=-8pt] (Ck) {};
			\node[zoki state, label=center:$q_{C_2}$] at ($(C1)!.3333!(Ck)$) (C2) {};
			\node[zoki state, label=center:$q_{C_3}$] at ($(C1)!.6666!(Ck)$) (Cdots) {};

			\path[zoki transition, green!50!olive]
			(C2.164) edge[to path={(\tikztostart) .. controls ++(165:.5) ..  (\tikztotarget) \tikztonodes}] (x1.330)
			(C2.164) edge[to path={(\tikztostart) .. controls ++(165:.5) ..  (\tikztotarget) \tikztonodes}] (x2.340)
			(C2.164) edge[to path={(\tikztostart) .. controls ++(165:.5) ..  (\tikztotarget) \tikztonodes}] (xn)
			(Cdots.190) edge[to path={(\tikztostart) .. controls ++(190:.5) ..  (\tikztotarget) \tikztonodes}] (x1c.30)
			(Cdots.190) edge[to path={(\tikztostart) .. controls ++(190:.5) ..  (\tikztotarget) \tikztonodes}] (x2c)
			(Ck) edge[to path={(\tikztostart) .. controls ++(180:.5) and  ++(320:2) ..  (\tikztotarget) \tikztonodes}] (x1c)
			(Ck) edge (xnc)
			(C1.195) edge (x2c)
			(C1.195) edge[bend right=20] (xnc)
			;
			\path[zoki transition, orange!50!yellow]
			(x1) edge (x2)
			(x1) edge[bend right] (x2c)
			(x1c) edge[bend left] (x2)
			(x1c) edge (x2c)
			(x2) edge (xn)
			(x2) edge[bend right] (xnc)
			(x2c) edge[bend left] (xn)
			(x2c) edge (xnc)
			(xn) edge[bend right=14] (x1)
			(xn) edge[to path={(\tikztostart) .. controls ++(165:8) and ++(160:5) ..  (\tikztotarget) \tikztonodes}] (x1c)
			(xnc) edge[bend left=14] (x1c)
			(xnc) edge[to path={(\tikztostart) .. controls ++(195:8) and ++(200:5) ..  (\tikztotarget) \tikztonodes}] (x1)
			;
			\path[zoki transition, blue!50!cyan]
			(q1) edge (x1)
			(q1) edge (x1c)
			(x1) edge (q2)
			(x1c) edge (q2)
			(q2) edge (x2)
			(q2) edge (x2c)
			(x2) edge (qdots)
			(x2c) edge (qdots)
			(qdots) edge (xn)
			(qdots) edge (xnc)
			(xn) edge (formula)
			(xnc) edge (formula)
			;
			\path[zoki transition, green!50!olive] 
			(formula) edge[shorten <=6pt] (C1)
			(formula) edge[shorten <=2pt] (C2)
			(formula) edge[shorten <=2pt] (Cdots)
			(formula) edge[shorten <=6pt] (Ck)
			;
			\path[zoki transition, red]
			(C1) edge (C2)
			(C2) edge (Cdots)
			(Cdots) edge (Ck)
			;
			\draw[zoki transition, rounded corners=4, red] (Ck.270) -- ++(270:.3) -| (q1);
			\node[zoki state, label=center:$q_3$, below right=of x2] (qdots) {};
		\end{tikzpicture}
	}
	\caption{The exploration game for QBF formula $\exists x_1{:}\forall x_2{:}\exists x_3{:}$$(C_1=x_2\lor x_3)$ $\land$ $(C_2=\lnot x_1 \lor \lnot x_2 \lor \lnot x_3)$ $\land$ $(C_3=x_1\lor x_2)$ $\land$ $(C_4=x_1\lor x_3)$, where colours encode edge availibility:
		red=$[0, k)$, blue=$[k, k+2n)$, green=$[k+2n, k+2n+1]$, yellow=$(k+2n+1, k+3n+1]$.}
	\label{fig:2p-temp-hard}
\end{figure}

\section{Symbolically represented Temporal Graphs}
\label{sec:compressed}
We now consider games on temporal graphs that are given in a succinct representation.
More precisely, we assume that the arena is given as the underlying (static) graph
and for every edge, the set of times at which it is available
is represented by a formula in the existential fragment of Presburger Arithmetic ($\exists$PA).
Recall that evaluating $\varphi(x)$ for a given $\exists$PA formula $\varphi$ and time $x\in\N$ is in $\NP$.
We call this representation of temporal graphs \emph{symbolic}.
Symbolic representations that are equally as powerful as $\exists$PA, such that they can represent semi-linear sets, include the likes of solving a union of linear equations and commutative context-free grammars \cite{CHISTIKOV2018147}.
Finding a satisfying valuation for both of these encodings of semi-linear sets has also been shown to be in $\NP$, much like $\exists$PA.
We use $\exists$PA as the formulae produced are concise and easy to understand after only a brief viewing.
There are other concise encodings (such as the inclusion of the universal fragment of PA) that could be used for a symbolic temporal graph, however the problem of finding membership or satisfiability of a given value may be too complex for the encoding to be efficient. For example, satisfiability with respects to the universal fragment of PA is exponential and not in $\NP$.


\begin{theorem}\label{lem:symbolic-hard}
	Solving one-player symbolic temporal reachability games is $\PSPACE$-hard.
\end{theorem}
\begin{proof}
	We reduce QSAT to one-player temporal reachability with symbolic time encoding.
	Consider a quantified Boolean formula $\varPhi = \exists x_1\forall x_2\ldots \exists x_n :\varphi$ where $\varphi$ is a propositional formula with variables in $X=\{x_1, \dots, x_n\}$.
	Without loss of generality assume that the quantifiers alternate strictly, starting and ending with existential quantifiers
	(introduce at most $n+1$ useless variables otherwise).
	We show how to construct the symbolic one-player game $G_\varPhi$ such that $\varPhi$ is satisfiable if and only if some target vertex of $G_\varPhi$ can be reached.
	We use elapse of time to encode valuations of variables of $\varPhi$.
	At a time $\theta\geq0$, since $\varphi$ is quantifier-free, determining whether the valuation $\nu_{\theta} \colon X \to \{0,1\}$ makes $\varPhi$ true can be checked with a single transition that is available exactly at time $\theta$.
	The difficulty arises from encoding of the ``adversarial'' behaviour of universal quantifiers.
	To do so, we leverage both the graph structure and the Presburger arithmetic.

	In our encoding, time is sectored into $n$ segments of 4 bits.
	The horizon of $G_\varPhi$ is thus $h(G_\varPhi)=2^{4n}$.
	For every $i\in\{1, \ldots, n\}$ and time $\theta \in \{0, ..., 2^{4n}-1\}$, we call $\alpha_i$, $\beta_i$, $\gamma_i$, and $\delta_i$ the four bits of the $i$th most significant sector, i.e., the bits at indices $(4n-1)-(4i-4)$, $(4n-1)-(4i-3)$, $(4n-1)-(4i-2)$ and $(4n-1)-(4i-1)$ in the binary expansion of $\theta$.
	In particular, $\alpha_{1}=4n-1$ and $\delta_{n}=0$ are respectively the most and the least significant bits.
	Intuitively, in the sequel we use $\beta_i$ to represent the chosen value of variable $x_i\in X$, we control overflows of $\beta_i$ with assumption on $\alpha_i$ and $\gamma_i$, and $\delta_i$ is an extra bit acting as buffer, to allow spending time without influencing the others.
	In the sequel, we define constrains on such bits.
	To do so, for all $k\in\{0, \ldots, 4n-1\}$ and all $v\in\{0, 1\}$, we require the $k$th least significant bit of a given time $\theta\in\{0, \ldots, 2^{4n}-1\}$ to be $v$. This can be expressed thanks to following Presburger formula, which states that the $k$th least significant bit in the binary expansion of $\theta$ equals $v$.
	$$\Psi_{k,v}(\theta) =
		\exists y_0:\ldots\exists y_{4n}:\exists v_0:\ldots\exists v_{4n-1}:
		\land \begin{cases}
			\bigwedge_{j=0}^{4n-1}(v_j = 0 \lor v_j=1) \hfill\quad \text{\textcolor{gray}{i.e., $v_j \in \{0,1\}$}}
			\smallskip\\
			\bigwedge_{j=0}^{4n-1}(y_{j} = 2y_{j+1} + v_j) \hfill\quad \text{\textcolor{gray}{i.e., $y_{j+1}$ is $\big\lfloor\frac{\theta}{j+1}\big\rfloor$}}
			\smallskip\\
			(y_{0}=\theta) \land (v_k=v) \hfill\quad \text{\textcolor{gray}{and $v_k$ is $\theta$'s $k$th bit}}
		\end{cases}
		$$
	In the following, we write $\xi_i=v$, where $\xi\in \{\alpha,\beta,\gamma,\delta\}$ and $v\in \{0,1\}$
	as a shorthand for the formula for checking the corresponding bits.

	\begin{figure}[t!]\centering
		\begin{minipage}{.35\linewidth}\centering
			\scalebox{1.15}{
				\begin{tikzpicture}[node distance=2cm]
					\tikzset{every label/.style={text=black, font=\small}}
					\node[zoki state,label=center:$\exists_i$, zoki initial, zoki initial style=dashed] (q1) {};
					\node[dashed, zoki state, right=of q1] (q) {};
					\node[opacity=0, zoki state, below=of q] (p) {};

					\path[zoki transition]
					(q1) edge [loop above] node[above] {$\alpha_i=0$} (q1)
					(q1) edge node[below] {$\begin{array}{c}\{\alpha_j=0\}_{j\in\N}\\\{\beta_j=0\}_{j>i}\\\{\gamma_j=0\}_{j\in\N}\\\{\delta_j=0\}_{j>i}\\\delta_i=1\\\end{array}$} (q)
					(q) edge[opacity=0, loop above] node[opacity=0] {$\gamma_i=0$} (q)
					(p) edge[opacity=0, loop left] node[opacity=0] {$\begin{array}{c}=0\\=0\end{array}$} (p)
					;
				\end{tikzpicture}
			}
		\end{minipage}
		\begin{minipage}{.64\linewidth}\centering
			\scalebox{1.15}{
				\begin{tikzpicture}[node distance=2cm]
					\tikzset{every label/.style={text=black, font=\small}}
					\node[zoki state, label=center:$\forall_{i}$, zoki initial, zoki initial style=dashed, zoki initial angle=180] (q2) {};
					\node[zoki state, label=center:$\forall_{i}^{c}$, above=of q2] (q2c) {};
					\node[zoki state, label=center:$\forall_{i}^{b}$, right=of q2c, xshift=-.5cm] (q2b) {};
					\node[zoki state, label=center:$\forall_{i}^{a}$, zoki initial, zoki initial style=dashed, zoki initial angle=-65, right=of q2b, xshift=-.5cm] (q2a) {};
					\node[dashed, zoki state, below=of q2a] (q3) {};
					\node[dashed, zoki state, above=of q2, xshift=-1cm] (q) {};

					\path[zoki transition]
					(q2) edge[loop below] node[left] {$\begin{array}{c}\alpha_i=0\\\delta_i=0\end{array}$} (q2)
					(q2) edge[pos=.75] node[above] {$\begin{array}{c}\{\alpha_j=0\}_{j\in\N}\\\{\beta_{j}=0\}_{j>i}\\\{\gamma_j=0\}_{j\in\N}\\\{\delta_j=0\}_{j>i}\\\delta_i=1\end{array}$} (q3)
					(q2) edge node [left]{$\alpha_i=1$} (q)
					(q2a) edge [loop above] node {$\gamma_{i} = 0$} (q2a)
					(q2a) edge node [above]{$\gamma_{i} = 1$} (q2b)
					(q2b) edge [loop above] node {$\gamma_{i} = 1$} (q2b)
					(q2b) edge node [above]{$\gamma_{i} = 0$} (q2c)
					(q2c) edge [loop above] node {$\gamma_{i} = 0$} (q2c)
					(q2c) edge node [right, xshift=-.2cm]{$\begin{array}{l}\{\alpha_j=0\}_{j\neq i}\\\{\beta_j=0\}_{j>i}\\\{\gamma_j=0\}_{j\in\N}\\\{\delta_j=0\}_{j\geq i}\end{array}$} (q2)
					;
				\end{tikzpicture}
			}
		\end{minipage}
		\vspace*{-.5cm}
		\caption{Gadgets to construct $G_\varPhi$}
		\label{QBF2onePlayerGame:quantifiers}
	\end{figure}

	In $G_\varPhi$, each quantifier will have a corresponding vertex.
	If $x_i\in X$ is existentially quantified, the vertex $q_i$ evaluates $x_i$ thanks to a self-loop which also set $\delta_i=1$. 
	In other words, \PONE can use the self-looping transition of $q_i$ for spending time (not more than $2^{\alpha_i}$ steps since $\alpha_{i}$ must be $0$) which in particular leaves him possibility to set the bit $\beta_{i}$ as an encoding the value of $x_i$.
	Additionally, the non-looping transition of $q_i$ is available only when the bit $\delta_{i}$ is $1$.
	See the left gadget of \figurename~\ref{QBF2onePlayerGame:quantifiers}.
	Observe that, if $q_i$ is entered at time $\theta$ where $\delta_i=0$, there are exactly two times $\theta_0, \theta_1$ when $q_i$ can be exited, and $\theta_0$ and $\theta_1$ only differ on the bit $\beta_i$.
	Formally, $\theta_0=\theta+2^{(4n-1)-(4i-1)}$, $\theta_1=\theta_0+2^{(4n-1)-(4i-3)}$.
	This is a direct consequence of the availability of both outgoing transitions of $q_i$ leaving only the bit $\beta_{i}$ as degree of freedom to \PONE.
	We shall prove that vertices encoding existential quantifiers are the only source of branching in $G_\varPhi$.
	If $x_i\in X$ is universally quantified, the self-loop of the vertex $q_i$ only serves to set $\delta_i=1$, the evaluation of $x_i$ being driven by vertices encoding existential quantifiers and auxiliary vertices $q_i^a$, $q_i^b$ and $q_i^c$.
	See the right gadget of \figurename~\ref{QBF2onePlayerGame:quantifiers}.
	The vertex $q_i$ admits two behaviours.
	If $q_i$ is entered at a time when $\alpha_i=0$, then $q_i$ exits toward the vertex encoding the next existential quantifier.
	Otherwise, $\alpha_i=1$ and then $q_i$ backtracks toward the auxiliary vertices of the previous universal quantifier.
	The elapsing of time while backtracking will set $\gamma_{i-2}$ to 1, and then sets it back to 0.
	If $\beta_{i-2}=0$, the side effect of backtracking is to switch $\beta_{i-2}$ from 0 to 1.
	Otherwise, $\beta_{i-2}=1$, and the side effect of backtracking is to switch $\alpha_{i-2}$ from 0 to 1, which triggers a further backtrack in $q_{i-2}$.
	To end the construction of $G_\varPhi$, we add two more vertices $q_\varphi$ (to check the valuation of $\varPhi$ encoded by elapse of time) and $q_\top$ (to end the cascade of backtrack).
	A complete picture is given when $\varPhi$ has 5 quantifiers in \figurename~\ref{QBF2onePlayerGame:example}.
	The quantifier-free Presburger formula $\widehat{\varphi}$ is defined as $\varphi$ where every positive occurrence $x_i$ in $\varphi$ is replaced by $\beta_i=1$ in $\widehat{\varphi}$ and every negative occurrence $\lnot x_i$ is replaced by $\beta_i=0$.
	The size of $G_\varPhi$, in particular the encoding of edges availability is discussed at the end of the proof.

	\begin{figure}[t!]\centering
		\scalebox{1}{
			\begin{tikzpicture}[node distance=1.92cm]
				\tikzset{every label/.style={text=black, font=\small}}
				\node[draw=red, zoki state, zoki initial, zoki initial angle=-110, label=center:$q_{1}$] (q1) {};
				\node[draw=blue, zoki state, label=center:$q_{2}$, right=of q1] (q2) {};
				\node[draw=orange!50!yellow, zoki state, label=center:$q_{3}$, right=of q2] (q3) {};
				\node[draw=green!50!olive, zoki state, label=center:$q_{4}$, right=of q3] (q4) {};
				\node[draw=violet!50!magenta, zoki state, label=center:$q_{5}$, right=of q4] (q5) {};
				\node[zoki state, label=center:$q_\varphi$, right=of q5] (formula) {};

				\node[draw=green!50!olive, zoki state, label=center:$q_{4}^b$, above=of q5, xshift=-.5cm] (q4b) {};
				\node[draw=green!50!olive, zoki state, label=center:$q_{4}^a$, right=of q4b, xshift=-.5cm] (q4a) {};
				\node[draw=green!50!olive, zoki state, label=center:$q_{4}^c$, above=of q4] (q4c) {};

				\node[draw=blue, zoki state, label=center:$q_{2}^b$, above=of q3, xshift=-.5cm] (q2b) {};
				\node[draw=blue, zoki state, label=center:$q_{2}^a$, right=of q2b, xshift=-.5cm] (q2a) {};
				\node[draw=blue, zoki state, label=center:$q_{2}^c$, above=of q2] (q2c) {};

				\node[zoki state, label=center:$q_\top$, above=of q2, xshift=-1cm] (end) {};

				\path[zoki transition]
				(q1) edge[draw=red, pos=.42] node[below] {$\begin{array}{c}\{\alpha_j=0\}_{j\in\N}\\\{\beta_j=0\}_{j>1}\\\{\gamma_j=0\}_{j\in\N}\\\{\delta_j=0\}_{j>1}\\\delta_1=1\end{array}$} (q2)
				(q2) edge[draw=blue, pos=.6] node[below] {$\begin{array}{c}\{\alpha_j=0\}_{j\in\N}\\\{\beta_{j}=0\}_{j>2}\\\{\gamma_j=0\}_{j\in\N}\\\{\delta_j=0\}_{j>2}\\\delta_2=1\end{array}$} (q3)
				(q3) edge[draw=orange!50!yellow, pos=.42] node[below] {$\begin{array}{c}\{\alpha_j=0\}_{j\in\N}\\\{\beta_j=0\}_{j>3}\\\{\gamma_j=0\}_{j\in\N}\\\{\delta_j=0\}_{j>3}\\\delta_3=1\end{array}$} (q4)
				(q4) edge[draw=green!50!olive, pos=.6] node[below] {$\begin{array}{c}\{\alpha_j=0\}_{j\in\N}\\\{\beta_{j}=0\}_{j>4}\\\{\gamma_j=0\}_{j\in\N}\\\{\delta_j=0\}_{j>4}\\\delta_4=1\end{array}$} (q5)
				(q5) edge[draw=violet!50!magenta, pos=.42] node[below] {$\begin{array}{c}\{\alpha_j=0\}_{j\in\N}\\\{\beta_j=0\}_{j>5}\\\{\gamma_j=0\}_{j\in\N}\\\{\delta_j=0\}_{j>5}\\\delta_5=1\end{array}$} (formula)
				(formula) edge node[below left] {$\widehat{\varphi}$} (q4a)
				(q4a) edge[draw=green!50!olive] node[above] {$\gamma_4=1$} (q4b)
				(q4b) edge[draw=green!50!olive] node[above] {$\gamma_4=0$} (q4c)
				(q4c) edge[draw=green!50!olive] node[right, xshift=-.2cm] {$\begin{array}{c}\{\alpha_j=0\}_{j\neq4}\\\{\beta_j=0\}_{j>4}\\\{\gamma_j=0\}_{j\in\N}\\\{\delta_j=0\}_{j\geq4}\end{array}$} (q4)
				(q4) edge[draw=green!50!olive] node [below left]{$\alpha_4=1$} (q2a)
				(q2a) edge[draw=blue] node[above] {$\gamma_2=1$} (q2b)
				(q2b) edge[draw=blue] node[above] {$\gamma_2=0$} (q2c)
				(q2c) edge[draw=blue] node[right, xshift=-.2cm] {$\begin{array}{c}\{\alpha_j=0\}_{j\neq2}\\\{\beta_j=0\}_{j>2}\\\{\gamma_j=0\}_{j\in\N}\\\{\delta_j=0\}_{j\geq2}\end{array}$} (q2)
				(q2) edge[draw=blue] node [below left]{$\alpha_2=1$} (end)
				(q2) edge [draw=blue, loop below] node {$\begin{array}{c}\alpha_2=0\\\delta_2=0\end{array}$} (q2)
				(q4) edge [draw=green!50!olive, loop below] node {$\begin{array}{c}\alpha_4=0\\\delta_4=0\end{array}$} (q4)
				(q1) edge [draw=red, loop above] node {$\alpha_1=0$} (q1)
				(q3) edge [draw=orange!50!yellow, loop above] node {$\alpha_3=0$} (q3)
				(q5) edge [draw=violet!50!magenta, loop above] node {$\alpha_5=0$} (q5)
				(q4a) edge [draw=green!50!olive, loop above] node {$\gamma_4=0$} (q4a)
				(q4b) edge [draw=green!50!olive, loop above] node {$\gamma_4=1$} (q4b)
				(q4c) edge [draw=green!50!olive, loop above] node {$\gamma_4=0$} (q4c)
				(q2a) edge [draw=blue, loop above] node {$\gamma_2=0$} (q2a)
				(q2b) edge [draw=blue, loop above] node {$\gamma_2=1$} (q2b)
				(q2c) edge [draw=blue, loop above] node {$\gamma_2=0$} (q2c)
				(end) edge [loop above] (end)
				;
			\end{tikzpicture}
		}
		\vspace*{-.5cm}
		\caption{Encoding of a formula with 5 quantifiers to a symbolic one-player game.}
		\label{QBF2onePlayerGame:example}
	\end{figure}

	Next, we prove that $\varPhi$ is satisfiable
	if and only if
	$G_\varPhi$ admits a path from $q_1$ to $q_\top$ composed of exactly $2^{4n}$ edges, thanks to the unrestricted self-loop on $q_{\top}$.
	The argument holds by induction on the odd number $n\geq1$ of quantifiers.
	In the base case, $n=1$ and $G_\varPhi$ has three vertices $q_1, q_{\varphi}$ and $q_\top$.
	We show that when $q_1$ is entered at time $0$, it can be exited only at time $\theta_0=1$ and $\theta_1=5$.
	By construction, $q_1$ can only be exited toward $q_{\varphi}$ that has a single outgoing edge toward $q_{\top}$.
	Hence, the edge from $q_{\varphi}$ to $q_{\top}$ can only be traversed at time $\theta_0+1=2$ or $\theta_1+1=6$ where $\beta_1=0$ and $\beta_1=1$ respectively.
	By definition of $\widehat{\varphi}$, the vertex $q_{\top}$ is reachable if and only if $\varPhi$ is satisfiable.

	Now, assume that $\varPhi$ has an odd number $n\geq3$ of quantifiers.
	Let $\varPhi = \exists x_1\forall x_2\varPhi'$ where $\varPhi' = \exists x_3\forall x_4\ldots \exists x_{n} :\varphi$.
	For all $v_1, v_2 \in \{\texttt{false}, \texttt{true}\}$, we denote $\varPhi'[x_1 \leftarrow v_1, x_2 \leftarrow v_2]$ the formula $\varPhi'$ where $x_1$ takes value $v_1$ and $x_2$ takes value $v_2$.
	Observe that $\varPhi'[x_1 \leftarrow v_1, x_2 \leftarrow v_2]$ is a quantified Boolean formula without free-variable.
	For all $v_1, v_2 \in \{\texttt{false}, \texttt{true}\}$, by induction hypothesis $G_{\varPhi'[x_1 \leftarrow v_1, x_2 \leftarrow v_2]}$ admits a path from $q'_1$ to $q'_\top$ composed of exactly $2^{4(n-2)}$ edges if and only if $\varPhi'[x_1 \leftarrow v_1, x_2 \leftarrow v_2]$ is satisfiable.
	Intuitively, we construct such a path of $G_{\varPhi}$ from $q_1$ to $q_{\top}$ based on a path of $G_{\varPhi'}$ from $q'_1$ to $q'_{\top}$.
	This is effective since $G_{\varPhi'[x_1 \leftarrow v_1, x_2 \leftarrow v_2]}$ is a subgraph of $G_{\varPhi[x_1 \leftarrow v_1, x_2 \leftarrow v_2]}$ where $q_3$ and $q_2^{a}$ in $G_{\varPhi[x_1 \leftarrow v_1, x_2 \leftarrow v_2]}$ match respectively with $q'_1$ and $q'_{\top}$ in $G_{\varPhi'[x_1 \leftarrow v_1, x_2 \leftarrow v_2]}$. 
	In fact, $G_{\varPhi'[x_1 \leftarrow v_1, x_2 \leftarrow v_2]}$ has an horizon of $2^{4(n-2)}$ because $\varPhi'$ has $n-2$ quantifiers,
	and $q_2^{a}$ has a self-loop available at all time in this horizon, that is, $G_{\varPhi'[x_1 \leftarrow v_1, x_2 \leftarrow v_2]}$ cannot modify the bits of bit sectors 1 and 2.
	Also, $q_3$ is entered only when the $4(n-2)$ least significant bits of the time are zero due to the availability of the edge from $q_2$ to $q_3$,
	and $q_2^{a}$ is leaved only if the $4(n-2)$ bits overflow due to as a direct consequence of the availability of the edge from $q_2^{a}$ to $q_2^{b}$.

	We show that when $q_1$ is entered at time $0$, it can be exited only at time
	$\theta_0=2^{4n-4}$, $\theta_1=2^{4n-4}+2^{4n-2}$.
	Then $q_2$ is visited and since in both cases $\alpha_2=0$, by construction, the vertex $q_3$ is visited at time
	$\theta_0+2^{4n-8}$ or $\theta_1+2^{4n-8}$ (where $\delta_2$ switched to $1$).
	Let $\theta'_0=\theta_0+2^{4n-8}$ where $\beta_1=0$ and $\theta'_1=\theta_1+2^{4n-8}$ where $\beta_1=1$.
	In both cases, $\beta_2=0$.
	For all $k\in\{0, 1\}$, by induction hypothesis, there is a path from $q_3$ at time $\theta'_k$ to $q_2^{a}$ at time $\theta'_k+2^{4(n-2)}-1$ if and only if $\varPhi'[x_1 \leftarrow k, x_2 \leftarrow 0]$ is satisfiable.
	If $\varPhi'[x_1 \leftarrow 0, x_2 \leftarrow 0]$ and $\varPhi'[x_1 \leftarrow 1, x_2 \leftarrow 0]$ are not satisfiable, then $G_{\varPhi}$ cannot reach $q_{\top}$ from $q_1$ and $\varPhi$ is not satisfiable.
	Otherwise, the two paths in $G_{\varPhi}$ that reach $q_2^{a}$ at $\theta'_0+2^{4(n-2)}-1$ and $\theta'_1+2^{4(n-2)}-1$ can be extended to reach $q_2$ at time
	$\theta_0+2^{4n-6}$ and $\theta_1+2^{4n-6}$ respectively (where $\beta_2$ switched to $1$).
	Since $\alpha_2=0$ in both cases, the vertex $q_3$ is visited at time
	$\theta_0+2^{4n-6}+2^{4n-8}$ and $\theta_1+2^{4n-6}+2^{4n-8}$ respectively (where switched to $\delta_2=1$).
	Let $\theta''_0=\theta_0+2^{4n-6}+2^{4n-8}$ where $\beta_1=0$ and $\theta''_1=\theta_1+2^{4n-6}+2^{4n-8}$ where $\beta_1=1$.
	For all $v_1, v_2 \in \{\texttt{false}, \texttt{true}\}$, by induction hypothesis, there is a path from $q_3$ at time $\theta''_k$ to $q_2^{a}$ at time $\theta_k+2^{4(n-2)}-1$ if and only if $\varPhi'[x_1 \leftarrow k, x_2 \leftarrow 1]$ is satisfiable.
	If $\varPhi'[x_1 \leftarrow 0, x_2 \leftarrow 1]$ and $\varPhi'[x_1 \leftarrow 1, x_2 \leftarrow 1]$ are not satisfiable, then
	$G_{\varPhi}$ cannot reach $q_{\top}$ and $\varPhi$ is not satisfiable.
	Otherwise, the two paths in $G_{\varPhi}$ that reach $q_2^{a}$ at $\theta''_0+2^{4(n-2)}-1$ and $\theta''_1+2^{4(n-2)}-1$ can be extended to  reach $q_2$ at time
	$\theta_0+2^{4n-5}$ and $\theta_1+2^{4n-5}$ respectively (where $\alpha_2$ switched to $1$).
	Since $\alpha_2=1$, the vertex $q_{\top}$ is visited at time
	$\theta_0+2^{4n-5}+1$ and $\theta_1+2^{4n-5}+1$ respectively.
	Hence, $G_{\varPhi}$ admit a path from $q_1$ to $q_{\top}$ if and only if there is $k\in\{0, 1\}$ such that $\varPhi'[x_1 \leftarrow k, x_2 \leftarrow 0]$ and $\varPhi'[x_1 \leftarrow k, x_2 \leftarrow 1]$ are satisfiable, i.e., if and only if $\varPhi$ is satisfiable.

	It is worth emphasizing that, when $G_\varPhi$ admits a path from $q_1$ to $q_{\top}$, it can be constructed with a memoryless strategy.
	In the above reduction, for each vertex $q_i$ that encodes an existential quantifier, the inductively constructed strategy aims in $q_i$ at evaluating $\beta_i$ such that the next time $q_{\varphi}$ is visited its outgoing edge will be available.
	Hence, in $q_i$, the strategy is solely based on the value $\beta_{i-1}, \ldots, \beta_1$.
	Since the vertices encoding existential quantifiers are the only source of branching in $G_\varPhi$, the strategy is memoryless.

	Finally, we show that the size of $G_\varPhi$ is polynomial in the size of $\varPhi$.
	Here after, the size of the formula corresponds to the number of symbols to write it.
	The size of a symbolic temporal graphs $G$, is defined by $|G| = |V| + \sum_{(u,v)\in V^2} |\edges(u, v)|$.
	The graph $G_\varPhi$ has at most $4n+2$ vertices, and all availability are expressible by a Presburger formula of polynomial size in $|\varPhi|$.
	Since the horizon of $G_\varPhi$ is $2^{4n}$, each edge availability can be expressed as a conjunction of at most $4n$ formulas of the form $\Psi_{k,v}$.
\end{proof}

\begin{corollary}
	Solving reachability games on one-player symbolic temporal graphs with at most $K$ temporal edges is $\Sigma_{m}^{P}$-hard, i.e, hard for the $m$-th level of the polynomial hierarchy, where $K\ge 2\lceil\frac{m}{2}\rceil + 9\lfloor\frac{m}{2}\rfloor+1$.
\end{corollary}
\begin{proof}
	In the proof of \cref{lem:symbolic-hard}, the number of quantifiers in $\varPhi$ determines the number of temporal edges used in the constructed symbolic temporal graph. The gadget to encode existential quantifiers uses two temporal edges and the one for universal quantifiers uses nine temporal edges, and there is one edge for the quantifier-free formula $\varphi$ (see \cref{QBF2onePlayerGame:quantifiers}).
	Note that if a quantifier block contains more than $m$ variables, we can blow up the number of bits polynomially by having bits $\beta_{i}^{j}$ for $1\le j \le m$, while keeping the number of temporal edges the same, thus obtaining a reduction from QSAT with a fixed number of quantifiers.
\end{proof}
A similar $\Pi_{m}^{P}$-hardness also holds for a different choice of $m$. Note that this is in contrast with the \PSPACE-hardness proof for two-player reachability games on symbolic temporal graphs, where the problem is $\PSPACE$-hard even with just one temporal edge~\cite[Theorem 1]{FH12}.
\medskip

A restriction that makes the problem easier is allowing \emph{waiting}. It is typical in study of temporal graphs to allow waiting for an arbitrary amount of time in temporal paths, i.e., instead of having to take an edge at every timestep, edges are traversed at increasing (but not necessarily consecutive) timesteps.
This corresponds to having a loop on all vertices that can be traversed at any time. In case of two-player games, this allows \PTWO to stall infinitely and therefore is not interesting.
The following remark states that the explorability on a one-player temporal arena with symbolic encoding becomes easy if waiting is allowed.
We show that this holds even for generalized reachability objectives.

\begin{theorem} \label{rem:1p-waiting}
	Assuming symbolic encodings, solving one-player games on temporal graphs
	with waiting is $\NP$-complete for reachability, explorability and generalized reachability.
\end{theorem}
\begin{proof}
	The $\NP$ lower bound follows from checking satisfiability of an existential Presburger formula, which is known to be $\NP$-complete~\cite{S84}.
	Given a $\exists$PA formula $\phi$ with no free variables, consider the symbolic temporal graph $\states=\{s,t\}$, such that
	$E(s,t)\eqdef\phi(x)$ is the only available edge, where $x$ is a variable that does not occur in $\phi$. The vertex $t$ is reachable from $s$ if, and only if,
	$\phi$ is true.
	For the upper bound, note that the length of a shortest witnessing path, measured by the number of edges traversed,
	is at most quadratic in the number of vertices of the graph.
	The $\NP$ upper bound can be shown by guessing the vertices visited along a play, together with the time they are visited first.
	The intermediate path between two such vertices can only visit linearly many vertices as no vertex is repeated in such a path.
	If a revisit is necessary, there is another path which waits instead.
	If a path does exist, then the time at which it is visited will at most be exponentially large and consequentially can be guessed in polynomial time.
	Therefore, by guessing the path and the times at which corresponding edges are taken, we obtain an $\NP$ algorithm.
\end{proof}

\begin{remark}
	In~\cite{FJ13}, Fearnley and Jurdzi\'nski introduce counter-stack automata to prove that reachability in two-clock timed automata is $\PSPACE$-complete.
	The model of counter-stack automata bears similarity with the technique used in proof of \cref{lem:symbolic-hard}.
	Indeed, both rely on a linear ordering of counters (being bits in our model) and that incrementation of a counter can only be performed when the value of all smaller counters is known.
	When elapse of time sets a bit to 1, the smaller bits are reset to 0.
	In~\cite{FJ13} counters can be reset, which correspond to reset of clocks.
	Our proof cannot be derived from~\cite{FJ13} because time cannot be reset in our setting.
	Vice versa, our proof uses existential Presburger formula with arbitrarily many variables which is not known to be expressible with a two-clock automata.
\end{remark}

\begin{theorem}\label{lem:symbolic-easy}
	Solving two-player symbolic temporal generalized reachability games is in $\EXP$.
	It is in $\PSPACE$ for one-player games.
\end{theorem}
\begin{proof}
	Consider a generalized reachability game $G$ on arena $\arena = (\states_1, \states_2, E)$
	and with targets $\mathcal{F}= F_1,\dots F_k$.
	Note that under symbolic encodings, temporal graphs are ultimately periodic, meaning that
	there are $b,p\in\N$ so that
	for every $i>b$, $i\in\edges(u,v)$ if, and only if, $i+p\in\edges(u,v)$.
	This is because the set of times at which any given edge is available is defined using a Presburger formula and thus semilinear \cite{GS66}.
	Moreover, the least common multiple of the periods of the individual edges is therefore a period of (all edges in) the temporal graph
	and we can bound $b$ and $p$ to be at most exponential in the size of the input.

	We can therefore construct an at most exponentially larger reachability game $G'$ 
	in which control states record the set of states (of $G$) that have already been visited, and the time up to and within the period.
	That is, define the arena $\arena' = (\states_1',\states_2',\edges')$ where
	$\states' = \states\times2^\states\times \{0,1,\dots, b+p-1\}$ and with edge relation $\edges'$ so that
	$(u,S,\temp)\step{}(v,S\cup\{v\},\temp')$
	if
	$\temp\in \edges(u,v)$ and
	either $\temp'=\temp+1$ or both $\temp= b+p-1 \text{ and } \temp'=b$.
	\PONE owns a state $(u,S,\temp)$ in $G'$ iff she owns $u$ in $G$.
	Finally, define the reachability target set in the constructed game $G'$ as
	\[F' = \{(v,S,\temp)\in \states': \forall F_j \in \mathcal{F}, \exists v \in F_j, \text{ such that } v\in S\}.\]
	Notice that any play from $(v,\emptyset,0)$ in $G'$ uniquely corresponds to a play from state $s$ at time $0$ in the original game $G$.
	The second components within control states in $G'$ are non-decreasing and record the set of states visited by the corresponding play in $G$.
	By definition of our reachability target $F'$, we see that
	a play is winning for \PONE in the constructed reachability game $G'$
	if, and only if, the corresponding play is winning for \PONE
	in the generalized reachability game $G$.

	Note that the size of $G'$ is $|\states|\cdot |2^{\states}|\cdot (b+p)$, where $b,p \le 2^{\?O(\abs{G})}$.
	The claim thus follows from the (known) facts that
	two-, and one-reachability games can be solved in polynomial time
	and logarithmic space, respectively.
\end{proof}

\begin{corollary}\label{lem:symbolic-all}
	Assuming symbolic encodings,
	solving one-player games on temporal graphs is $\PSPACE$-complete
	for reachability, explorability and generalized reachability.
	Solving two-player games on temporal graphs is $\PSPACE$-complete for reachability and $\PSPACE$-hard and in
	$\EXP$ for explorability and generalized reachability.
\end{corollary}

\section{Conclusion}
\label{sec:conclusion}
We study the complexities of solving explorability games on temporal graphs,
splitting into several cases based on the way temporal edges are defined and whether one or two-players are active.

First, we transfer results for games on static graphs and show that there, solving explorability games is $\NL$-complete for the one-player and $\P$-complete for the two-player case.

On explicitly represented temporal graphs in which waiting is permitted, it was known that checking the existence of an exploring path, i.e.~the solving one-player explorability games is $\NP$-complete \cite{Michail2015}. We show here that even for the more general case where waiting is not always allowed the problem remains in $\NP$.
We further show that solving two-player games on explicitly represented temporal graphs is $\PSPACE$-complete
for reachability, explorability, and generalized reachability conditions.
The existing lower bound for reachability \cite{ABT24} crucially relies on binary encodings and the presence of an antagonistic player.

The main technical contribution of this work is a strong lower bound for symbolically represented temporal graphs:
We show that already single-player reachability (and thus explorability) games are $\PSPACE$-hard. Our construction uses specially crafted Presburger constraints to ensure that all possible choices of an ``opponent'' in a QBF game are verified. 
Towards upper bounds, we show that even two-player generalized reachability games are solvable in deterministic exponential time, which is better than the exponential \emph{space} procedure one gets for free (by encoding time into the state space) but not quite matching our lower bound. We conjecture that indeed, solving two-player games for each of these three objectives remains in $\PSPACE$, mainly because raising the lower bound remained elusive. To show $\EXP$-hardness one needs to encode non-trivial and recoverable information into timestamps.
Our attempts to reduce from CTL Satisfiability or Countdown Games so far required to either construct Presburger constraints to compare bits that are exponentially far apart in the binary expansion of a given free variable (time), or to construct gadgets that ``copy'' bitstrings from lower to higher significant bit positions. It is worth noting that any such construction must only be correct for explorability but not for the simpler reachability conditions (which are $\PSPACE$-complete \cite{ABT24}).

On the other hand, showing $\PSPACE$-membership would require either
find a polynomially bounded untimed game gadgets to evaluate any given $\exists$PA formula
or to encode the antagonistic behaviours arithmetically.

\bibliography{bib/journals,bib/conferences,bib/references.clean}

\appendix
\end{document}